\PassOptionsToPackage{sort&compress}{natbib}  
\documentclass{elsarticle}

\usepackage{lineno,hyperref}
\modulolinenumbers[5]

\usepackage[american]{babel}	
\usepackage[utf8]{inputenc}		
\usepackage[T1]{fontenc}		
\usepackage[english=american]{csquotes}	
\usepackage{graphicx}			
\usepackage{epstopdf}
\usepackage{subfig}				
\usepackage{marvosym}			
\usepackage{amsmath}			
\usepackage{amssymb}
\usepackage{amsthm}
\usepackage{mathtools}
\usepackage{color}
\usepackage[super]{nth}

\usepackage{natbib}

\usepackage{hyperref}

\usepackage[ruled,vlined]{algorithm2e}

\usepackage{calligra}
\DeclareMathAlphabet{\mathcalligra}{T1}{calligra}{m}{n} 

\theoremstyle{definition}
\newtheorem{definition}{Definition}

\newtheorem{remark}{Remark}

\theoremstyle{plain}
\newtheorem{lemma}{Lemma}
\newtheorem{theorem}{Theorem}
\newtheorem{proposition}{Proposition}
\newtheorem{corollary}{Corollary}

\newcommand{\red}{\textcolor{black}}

\begin{document}

\begin{frontmatter}
\title{ Extremal values of the Sackin tree balance index}

\author{Mareike Fischer\corref{cor1}}
\ead{email@mareikefischer.de}
\cortext[cor1]{Corresponding author}

\address{Institute of Mathematics and Computer Science, Greifswald University, Greifswald, Germany}

\begin{abstract} Tree balance plays an important role in different research areas like theoretical computer science and mathematical phylogenetics. For example, it has long been known that under the Yule model, a pure birth process, imbalanced trees are more likely than balanced ones. \red{Also, concerning ordered search trees, more balanced ones allow for more efficient data structuring than imbalanced ones.} Therefore, different methods to measure the balance of trees were introduced. The Sackin index is one of the most frequently used measures for this purpose. In many contexts, statements about the minimal and maximal values of this index have been discussed, but formal proofs have \red{only been provided for some of them, and only in the context of ordered binary (search) trees, not for general rooted trees}. Moreover, while the number of trees with maximal Sackin index as well as the number of trees with minimal Sackin index when the number of leaves is a power of 2 are relatively easy to understand, the number of trees with minimal Sackin index for all other numbers of leaves has been completely unknown. In this manuscript, we \red{extend the findings on trees with minimal and maximal Sackin indices from the literature on ordered trees and subsequently use our results to} provide formulas to explicitly calculate the numbers of such trees. \red{We also extend previous studies by analyzing the case when the underlying trees need not be binary. Finally, we use our results to contribute both to the phylogenetic as well as the computer scientific literature by using the new findings on Sackin minimal and maximal trees in order to derive formulas to calculate the number of both minimal and maximal phylogenetic trees as well as minimal and maximal ordered trees both in the binary and non-binary settings. All our results have been implemented in the Mathematica package SackinMinimizer, which has been made publicly available.}
\end{abstract}

\end{frontmatter}


\section{Introduction}

Rooted trees, and binary ones in particular, play a fundamental role in many sciences as they can be used as a basis for search algorithms \cite{knuth1,knuth3} as well as, amongst others, as a model for evolution \cite{felsenstein_2004,semple_steel_2003}. In many cases where these trees occur, probability distributions are not always uniform concerning the degree of tree balance -- for instance, the Yule model in phylogenetics, which is a pure birth process, has long been known to lead to more imbalanced trees. A simple example is depicted in Figure \ref{yule}, where it can be seen that if all leaves of a tree with three leaves are equally likely to give rise to a new leaf, then two of them lead to the same (`imbalanced') tree, whereas the other possible tree (the `balanced' one) occurs only once. So in order to understand such processes and their possible bias towards imbalanced (or, in other cases, balanced) trees, one has to be able to classify the degree of balance in more detail than just in a binary way (`balanced' versus `imbalanced'). Therefore, various balance indices were introduced and have been used over the years, e.g.  \citep{blum,colless,cophenetic,sackin,steel2016}. One of the most frequently used and discussed such indices is the Sackin index \cite{sackin}.

\begin{figure}
	\centering
	\includegraphics[scale=.3]{yule.eps}
	\caption{The Yule process splits one leaf at a time uniformly at random to form a so-called cherry. It can be easily seen that this leads to a tree shape bias already when there are $n=4$ leaves. This is due to the fact that the only rooted binary tree on three leaves has two leaves that give rise to the tree on the left, which is considered `imbalanced' (it is the so-called caterpillar tree $T^{cat}_4$ on 4 leaves), whereas only one leaf leads to the so-called fully balanced tree $T_2^{bal}$ of height 2.
	}
	\label{yule}
\end{figure}

This index has been observed to have some very nice properties \red{for binary trees} -- for instance, it has been stated that its maximum is achieved by the caterpillar tree (the unique  tree with only one `cherry', i.e. with only one internal node whose two descendants are both leaves) and that, whenever the number $n$ of leaves equals a power of 2, i.e. when $n=2^k$ for some $k$, the minimum is achieved by the so-called fully balanced tree of height $k$, i.e. the tree in which all leaves have distance $k$ to the root  \cite{shao}. \red{While these statements can be found in the phylogenetic literature, they are typically stated there both without rigorous proofs and without references.}

\red{However, these statements are actually known, at least to some extent, and have been proven in a totally different context -- namely, in theoretical computer science. They actually trace back to the famous book series \enquote{The art of computer programming} by Knuth \cite{knuth1}, the first edition of which appeared already in the late 1960s. Knuth and other computer scientists do not use the term Sackin index; they use the term \emph{total external path length} instead, and they do not consider general rooted binary trees but ordered binary search trees (i.e. binary trees with a distinction between left and right), which may have contributed to the fact that their proofs are widely unknown in the phylogenetic community, even though the characterizations given in \cite{knuth1} actually also apply to non-ordered trees. }

\red{In this manuscript, as a first step we will present some known results and transfer them from ordered search trees to general rooted binary trees. We will thus show that the Sackin index indeed has the above mentioned properties, which are desirable for a good tree balance index, as the caterpillar tree is normally perceived as very `imbalanced', whereas the fully balanced tree is normally referred to as very `balanced' (which explains its name).} 

\red{For non-binary trees, however, even less has been proven in the literature, even though this case is less involved and some statements can be found in some manuscripts, even if they are to the best of our knowledge not stated anywhere with rigorous proofs. For instance, some authors mention that the so-called star tree, i.e. the unique rooted tree with only one inner node, is the unique Sackin minimal tree \cite{spanierNeu}. We present a concise rigorous proof for this as well as for the fact that even in the setting where the tree need not be binary, the binary caterpillar is still the unique most imbalanced tree concerning the Sackin index.} 

The idea of the Sackin index is to assign a small number to trees that are perceived as balanced and a high number to more imbalanced trees -- i.e. the higher the Sackin index, the more imbalanced the tree. So while some statements on the maximum of the Sackin index and its minimum in the special case $n=2^k$ can already be found in the literature, even if without proofs, little has been known about trees with the minimum Sackin index for $n\neq 2^k$ \red{-- in particular, so far no formula has been known to count these trees. While it has long been known that the binary tree achieving this minimum in such cases need not be unique, the number of most balanced or most imbalanced binary trees has never been formally investigated; neither in the phylogenetic nor in the computer scientific  literature.} 

\red{It is the first aim the present manuscript to give an overview of the different versions of the Sackin index that can be found in the literature and to link them to the notion of \emph{total external path length} in computer science. We then fully characterize both for binary and non-binary trees all Sackin minimal and maximal trees as well as the corresponding minimal and maximal Sackin values. We also point out which of these characterizations and proofs are new to the literature and which ones can be traced back, for instance, to theoretical computer science. We then turn our attention both to phylogenetic as well as to ordered trees and generalize our findings to these settings.}

\red{Ultimately, we use our characterizations to count all Sackin minimal and maximal trees -- in the binary and non-binary settings, in the phylogenetic setting, and in the ordered setting. Some of these enumerations are recursive, while others are explicit, and some of them are even new to the Online Encyclopedia of Integer Sequences \cite{OEIS}, i.e. they have not occurred in other combinatorial contexts yet, whereas some can be linked to other research areas.}

\red{In the last section of the present manuscript, we give a brief discussion and point out some areas for future research.}

\red{We have implemented the results presented in this manuscript in a Mathematica \cite{Mathematica} package called SackinMinimizer and made this package publicly available \cite{mathematicapackage}.} 

\section{Preliminaries}\label{sec_prelim}

Before we can start to discuss the Sackin tree balance index, we first need to introduce all concepts used in this manuscript. We start with trees: {\em Trees} are connected, acyclic graphs with node set $V$ and edge set $E$. We use $V^1$ in order to denote the set of {\em leaves} of a tree, i.e. the set of nodes of degree at most 1. All nodes $v$ that are not leaves, i.e. $v\in V\setminus V^1$, are called {\em internal nodes}. The set of internal nodes of a tree $T$ will be denoted by $\mathring{V}(T)$, or, whenever there is no ambiguity, simply by $\mathring{V}$. 

All trees $T$ in this manuscript are assumed to be {\em rooted}, i.e. they have a designated inner node that is called {\em root}. \red{Apart from the root, no vertices of degree 2 are allowed.} \red{Unless explicitly stated otherwise, we consider trees to be {\em binary}\footnote{Note that our concept of binary trees corresponds to the one in mathematical phylogenetics, where only one node of degree 2 is allowed, namely the root. This is, however, in contrast to computer science, where sometimes binary trees are allowed to have multiple nodes of degree 2.}}, i.e. if these trees have an internal node at all, they have one root node $\rho$ of degree 2 and all other internal nodes have degree 3. The only rooted binary tree which does not have an internal node is the tree that consists of only one node and no edge -- in this special case, the only node is for technical reasons at the same time defined to be the root and the only leaf of the tree, so it is the only case where the root is not an internal node. \red{Please note that by a slight abuse of notation, which is common in mathematical phylogenetics, the set of non-binary trees contains also the binary ones -- non-binary thus is meant in the sense of \enquote{not necessarily binary}. So if we discuss Sackin minimal and maximal non-binary trees in Section \ref{sec_nonbin}, this optimization includes also binary trees.}

Furthermore, for technical reasons all tree edges in this manuscript are implicitly assumed to be {\em directed} from the root to the leaves. Thus, for an edge $e=(u,v)$ of $T$, it makes sense to refer to $u$ as the {\em direct ancestor } or {\em parent } of $v$ (and $v$ as the {\em direct descendant} or {\em child} of $u$). More generally, when there is a directed path from $\rho$ to $v$ employing $u$, $u$ is called {\em ancestor} of $v$ (and $v$ descendant of $u$). Two leaves $v$ and $w$ are said to form a {\em cherry}, denoted by $[v,w]$, if $v$ and $w$ have the same parent, i.e. if there exists an internal node $u$ in $V$ such that $(u,v)$ and $(u,w)$ are edges in $E$. Note that every rooted tree with at least 2 leaves has at least one cherry.

Let $T$ be a rooted tree with root $\rho$, and let $x \in V^1$ be a leaf of $T$. Then we denote by $\delta_x$ the {\em depth} of $x$ in $T$, which is the number of edges on the unique shortest path from $\rho$ to $x$. Then, the {\em height} of $T$ is defined as $h(T)=\max\limits_{x \in V^1} \delta_x$, i.e. as the maximum of these distances. Note that whenever a leaf $v$ has maximal depth, i.e. whenever $h(T)=\delta_v$, $v$ is element of a cherry. This is due to the fact that if another direct descendant\footnote{\red{Note that in a binary tree, there is precisely one other such descendant; in a non-binary tree, there may be even more than that.}}, say $w$, of the parent of $v$, say $u$, was not a leaf but an internal node, it would have descending nodes of a greater depth than $\delta_v=\delta_w$, which would contradict the maximality of $\delta_v$. This is why in this manuscript, instead of considering both $v$ and $w$ separately as leaves of maximal depth, we sometimes refer to a cherry $[v,w]$ as {\em cherry of maximal depth}.

Moreover, recall that a rooted binary tree $T$ can be decomposed into its two maximal pending subtrees $T_a$ and $T_b$ rooted at the direct descendants $a$ and $b$ of $\rho$, and we denote this by $T=(T_a,T_b)$. \red{In the unrooted setting, this so-called \emph{standard decomposition} of $T$ into the maximal subtrees pending from the root may result in more than two trees.}

\red{As we will later on transfer some findings from graph theoretical trees back to research areas like mathematical phylogenetics as well as theoretical computer science, we need two more notions of trees: \begin{itemize}\item A rooted (binary or non-binary) \emph{phylogenetic} $X$-tree is a rooted (binary or non-binary) tree without degree-2 vertices other than the root and whose leaves are bijectively labeled by some taxon set $X$. We may typically assume without loss of generality that $X=\{1,\ldots,n\}$, in which case we simply talk about phylogenetic trees with $n$ leaves rather than phylogenetic $X$-trees. \item A rooted \emph{ordered} tree is a rooted (binary or non-binary) tree without degree-2 vertices other than the root and which has the property that for all inner nodes $u$, its descendants are ordered, e.g. from \enquote{left} to \enquote{right}. So in the ordered setting, if we for instance swap the left and the right maximal pending subtrees of $T_2$ in Figure \ref{6taxabalanced}, this would change the ordered version of $T_2$, whereas $T_2$ as a mere tree would remain unchanged.  \end{itemize} }

Last but not least, we want to introduce \red{three} particular trees which play a crucial role in this manuscript, namely the so-called {\em caterpillar tree} $T^{cat}_n$, the so-called {\em fully balanced tree} $T^{bal}_k$ \red{and the {\em star tree } $T^{star}_n$}, respectively. $T^{cat}_n$ denotes the unique rooted binary tree with $n$ leaves that has only one cherry, while $T^{bal}_k$ denotes the unique rooted binary tree with $n=2^k$ leaves in which all leaves have depth precisely $k$. \red{$T^{star}_n$ denotes the unique tree on $n$ leaves in which all leaves are directly pending from the root, i.e. in which all leaves have distance 1 to the root.} Whenever there is no ambiguity concerning $n$ or $k$, we also write $T^{cat}$, $T^{bal}$ and $T^{star}$ instead of $T^{cat}_n$, $T^{bal}_k$ and $T^{star}_n$. $T^{cat}_4$ and $T^{bal}_2$ are depicted in the bottom row of Figure \ref{yule}. Note that without loss of generality, the unique rooted binary tree with only one leaf, which consists of only one node and no edges, is defined to be $T^{cat}_1$, and it is thus the only caterpillar tree which does not contain a cherry. This tree is at the same time equal to $T^{bal}_0$, i.e. the fully balanced tree of height 0 and it is also equal to $T^{star}_1$. This technicality enables inductive proofs concerning $T^{cat} $, $T^{bal}$ and $T^{star}_1$to start at $n=1$. 

We are now in a position to define the central concept of this manuscript, namely the Sackin index. As there are four different versions of this index to be found in the literature, we will define all of them first and subsequently investigate their respective relationships. Note, however, that we focus on the first three definitions in the present manuscript, which can be shown to be equivalent. Therefore, we will not refer to the last definition as Sackin index, but give it a modified name instead.

\begin{definition}\label{sackin} \cite{cophenetic,steel2016}
The {\em Sackin index} of a rooted tree $T$ is defined as $\mathcal{S}(T)=\sum\limits_{u \in \mathring{V}(T)}n_u$, where $n_u$ denotes the number of leaves in the subtree of $T$ rooted at $u$. 
\end{definition}

Note that in the following, whenever we have for two trees $T_1$ and $T_2$ that $\mathcal{S}(T_1) < \mathcal{S}(T_2)$, then $T_1$ is called {\em more balanced} than $T_2$.

\begin{definition}\label{sackinalternative} \cite{cophenetic,steel2016}
The {\em Sackin index} of a rooted tree $T$ is defined as $\bar{\mathcal{S}}(T)=\sum\limits_{x \in V^1(T)}\delta_x$.
\end{definition}

\red{\begin{remark} The notion of the Sackin index given by Definition \ref{sackinalternative}, i.e. the sum of all path lengths from the root to any leaf (\enquote{external node}) of the tree, is known as \emph{total external path length} in computer science \cite{knuth1, cameron, klein,nievergelt}, where it is typically applied to ordered search trees. We will consider these more in-depth in Section \ref{sec_ordered}.
\end{remark}}

\begin{definition}\citep[(3.10)]{steel2016}\label{mike}
The {\em Sackin index} of a rooted tree $T$ with root $\rho$ and vertex set $V(T)$ is defined as $\widetilde{\mathcal{S}}(T)=\sum\limits_{u \in V(T)\setminus\{\rho\}} n_u$, where $n_u$ denotes the number of leaves in the subtree of $T$ rooted at $u$.  
\end{definition}

Note that in the original paper by Sackin \citep{sackin}, in fact no index is defined at all. Instead, a sequence $b$ of  leaf depths is defined, which implies that Definition \ref{sackinalternative} is probably most closely related to what Sackin originally intended. However, we will show in Lemma \ref{equidefs} that the first three definitions are in fact equivalent, which does not hold for the following definition.

\begin{definition}\label{noah}\cite{noah}
The {\em normalized Sackin index} of a rooted tree $T$ with $n$ leaves is defined as $\widehat{\mathcal{S}}(T)=\frac{1}{n}\sum\limits_{x \in V^1(T)}\delta_x$, where $\delta_x$ denotes the depth of leaf $x$. 
\end{definition}

We now state a first lemma, which has already been partially stated (albeit without proof) in the literature  \cite{cophenetic}. 
 
 \begin{lemma}\label{equidefs} Definitions \ref{sackin}, \ref{sackinalternative} and \ref{mike} are equivalent, i.e. for any rooted binary tree $T$, we have $\mathcal{S}(T)=\bar{\mathcal{S}}(T)=\widetilde{\mathcal{S}}(T)$.
 \end{lemma}
 
\red{ \begin{remark} Before we prove Lemma \ref{equidefs}, we want to state that this lemma is actually crucial in order to transfer findings on $\bar{\mathcal{S}}$, which can be derived from theoretical computer science, to the setting of mathematical phylogenetics, for instance, where the definition given by $\mathcal{S}$ is more common. \end{remark}}

 \begin{proof} 
 
We first prove $\mathcal{S}(T)=\bar{\mathcal{S}}(T)$. Therefore, consider $\bar{\mathcal{S}}(T)=\sum\limits_{x \in V^1(T)}\delta_x$. Note that while $\delta_x$ by definition denotes the number of edges separating leaf $x$ from the root $\rho$ of $T$, this is equivalent to the number of internal nodes on the path from $\rho$ to $x$ (including $\rho$). This leads to: $$\bar{\mathcal{S}}(T)=\sum\limits_{x \in V^1(T)}\delta_x= \sum\limits_{x \in V^1(T)} |\{u \in \mathring{V}(T): \mbox{$u$ is an ancestor of $x$}\}|$$ 

$$ = |\{(x,u): \mbox{ $x \in V^1(T)$, $u \in \mathring{V}(T)$ and $u$ is an ancestor of $x$}\}|$$
$$ = |\{(x,u): \mbox{ $x \in V^1(T)$, $u \in \mathring{V}(T)$ and $x$ is a descendant of $u$}\}|$$
$$= \sum\limits_{u \in \mathring{V}(T)} |\{x \in V^1(T): \mbox{$x$ is a descendant of $u$}\}| =\sum\limits_{u \in \mathring{V}(T)} n_u=\mathcal{S}(T).$$
So now we have $\mathcal{S}(T)=\bar{\mathcal{S}}(T)$, and next we show that $\mathcal{S}(T)=\widetilde{\mathcal{S}}(T)$. We have $$\mathcal{S}(T) = \sum\limits_{u \in \mathring{V}(T)}n_u = \sum\limits_{u \in V(T)}n_u - \sum\limits_{x \in V^1(T)}n_x = \left(\sum\limits_{u \in V(T)}n_u \right)-n.$$ The latter equality is due to the fact that the only leaf belonging to a subtree rooted at a leaf is the leaf itself, so as we have $n$ leaves, this gives $n$ summands that contribute 1 to the sum. Now recall that $n_\rho = n$, so this leads to 
$$\mathcal{S}(T) =\left( \sum\limits_{u \in V(T)}n_u\right) -n_\rho = \sum\limits_{u \in V(T)\setminus\{\rho\}}n_u=\widetilde{\mathcal{S}}(T).$$ This completes the proof.

 \end{proof}
 
So because Definitions \ref{sackin}, \ref{sackinalternative} and \ref{mike} are equivalent, we do not have to distinguish between them and will use them interchangeably. In fact, we will focus in this manuscript mainly on second definition. 

The normalized Sackin index, however, is a modification of the Sackin index whenever trees with different numbers of leaves are considered, because the ranking induced by the normalized Sackin index can even reverse the ranking induced by the Sackin index. For instance, consider the two trees $T_1=T_{37}^{cat}$, i.e.  the caterpillar tree with 37 leaves, and $T_2=T_{9}^{bal}$ the fully balanced tree with $2^9=512$ leaves. Then, it can easily be verified that we have $\mathcal{S}(T_1)=702 <  4608=\mathcal{S}(T_2)$, but $\widehat{\mathcal{S}}(T_1)\approx 18.97 > 9 =\widehat{\mathcal{S}}(T_2)$.\footnote{Note that these numbers can also be verified later on by using Theorems \ref{thm_sackinCat} and \ref{thm_twodepthsminNEW}.} In fact, this ranking modification is no artifact but the very purpose of the normalization: The effect that many leaves automatically may lead to more `imbalance' shall be eliminated. So in fact, $\mathcal{S}$ and $\widehat{\mathcal{S}}$ can be very different -- but only when different leaf numbers are considered! As long as $n$ is fixed, the induced rankings of the two indices are of course equivalent, and in this case, $\widehat{\mathcal{S}}$ is just $\mathcal{S}$ divided by the constant factor $n$. So when we discuss for instance the question how many trees with $n$ leaves exist that have maximal or minimal Sackin index, the answers for $\mathcal{S}$ and $\widehat{\mathcal{S}}$ will be the same. 

Therefore, as this is sufficient for the numbers of minima and maxima, we focus in this manuscript on Definition \ref{sackinalternative}.

\par\vspace{0.5cm}
\color{black}
However, before we can proceed with the results concerning the Sackin index, we need to introduce one more concept, which is closely related to Definition \ref{sackinalternative}.

\begin{definition}\label{def_int} Let $T$ be a rooted tree with $n$ leaves. Then, we set $$i(T):=\sum\limits_{v \in \mathring{V}(T)} \delta_v$$. $i(T)$ is then called the \emph{total internal path length} of $T$. 
\end{definition}

The following simple lemma has been stated already in \cite[p.400]{knuth1} (using different notation), but without a formal proof, so that we briefly prove it before we turn our attention to the extremal values of the Sackin index.

\begin{lemma} \label{lem_inout}
Let $T$ be a rooted binary tree with $n$ leaves. Then, we have $\mathcal{S}(T) = i(T) + 2(n-1)$.
\end{lemma}

\begin{proof} We prove the statement by induction on $n$. If $n=1$, $T$ consists only of one node, which is a leaf and has depth 0. So $\mathcal{S}(T)=0$. In this case, there is no inner vertex, so the sum in the definition of $i(T)$ is empty, so $i(T)=0$. So indeed we have $\mathcal{S}(T) =0= i(T)=i(T) + 2(1-1)=i(T) + 2(n-1)$, which completes the base case of the induction. Next, assume the statement holds for all trees with up to $n$ leaves and consider $T$ with $n+1\geq 2$ leaves. As every rooted binary tree on at least two leaves has at least one cherry, we can construct a tree $T'$ by deleting the leaves of a cherry $[x,y]$ of $T$, which turns its parent, say $v$, into a leaf. By the inductive hypothesis, we then have $\mathcal{S}(T') = i(T') + 2(n-1)$. Moreover, by the construction of $T'$, we have $i(T')=i(T)-\delta_v$, because $v$ is a leaf in $T'$ but an inner vertex in $T$. Moreover, we have $\mathcal{S}(T) =\mathcal{S}(T') +\delta_x+\delta_y-\delta_v$, as $x$ and $y$ are leaves in $T$ but not in $T'$, and as $v$ is a leaf in $T'$ but an inner node in $T$. Using $\delta_x=\delta_y=\delta_v+1$, this immediately shows that $\mathcal{S}(T) =\mathcal{S}(T') +\delta_v+2$. Using the inductive hypothesis, this leads to $\mathcal{S}(T) = i(T') + 2(n-1)+\delta_v+2=(i(T)-\delta_v)+2(n-1)+\delta_v+2=i(T)+2n=i(T)+2((n+1)-1)$. This completes the proof.
\end{proof}

\color{black}
\section{Results}
 
\subsection{Minimally and maximally balanced binary trees} \label{sec_binary}
\red{It is one of the main aims of this manuscript to count trees with minimal and maximal Sackin index, respectively, and in order to do so, we need to characterize such trees. In the present section, we focus on binary trees and for these, we will start with the easier case, which is the maximum. Afterwards we will consider the more involved and therefore more interesting case of the minimum.} 

\subsubsection{Maximally imbalanced binary trees / Minimally balanced binary trees} \label{sec_binmax}\par

In this subsection, we prove that $\mathcal{S}(T)$ is uniquely maximized by $T=T_n^{cat}$ for binary trees, and we also explicitly state the value of $\mathcal{S}(T_n^{cat})$, which is maximal. So we will show that for all values of $n$, $T^{cat}_n$ is the unique tree maximizing $\mathcal{S}$, i.e. the unique most imbalanced tree. This result has been stated in the literature before, e.g. in \cite{cophenetic} \red{as well as for ordered trees in \cite[p. 400]{knuth1}}, but so far, a formal proof has not been stated anywhere.

%
\begin{theorem} \label{thm_sackinCat}Let $T$ be a rooted binary tree with $n$ leaves and maximal Sackin index, i.e. we have $\mathcal{S}(T)\geq \mathcal{S}(T')$ for all rooted binary trees $T'$ on $n$ leaves. Then, $T$ is a caterpillar, i.e. $T=T_n^{cat}$. In other words, $T_n^{cat}$ is the unique binary tree maximizing $\mathcal{S}$. Moreover, we have $\mathcal{S}(T_n^{cat})=\frac{n\cdot(n+1)}{2}-1$.
\end{theorem}

\begin{proof} \begin{enumerate}
\item We begin by proving that the caterpillar tree is the unique binary tree maximizing $\mathcal{S}$. Assume this is not the case, i.e. assume there is a tree $T$ that is \emph{not} a caterpillar and that maximizes $\mathcal{S}$, i.e. $\mathcal{S}(T)\geq \mathcal{S}(T')$ for all rooted binary trees $T'$ on $n$ leaves. As $T$ is not a caterpillar, $T$ by definition has at least two cherries. Let $[x,y]$ be a cherry of $T$ of maximum depth, and let $[u,v]$ with parent $w$ be any other cherry of $T$. Then we have: $\delta_x=\delta_y\geq \delta_u=\delta_v=\delta_w+1$. Now we construct a tree $\widetilde{T}$ from $T$ as follows: We delete leaves $u$ and $v$ and the edges leading to these leaves, respectively, and we attach two new leaves $u'$ and $v'$ to leaf $x$. This implies that in total, $\widetilde{T}$ has three leaves that $T$ does not have, namely $u'$ and $v'$, but also $w$, which is not a leaf in $T$ but which is a leaf in $\widetilde{T}$, as $u$ and $v$ have been deleted. However, $T$ also has three leaves that $\widetilde{T}$ does not have, namely $u$ and $v$, which have been deleted to get to  $\widetilde{T}$, but also $x$, which is an internal node in $\widetilde{T}$ as $u'$ and $v'$ have been attached to it. In summary, this leads to:

$$\mathcal{S}(\widetilde{T})=\mathcal{S}(T)-\delta_u-\delta_v-\delta_x+\delta_{u'}+ \delta_{v'}+\delta_w.$$

Using $\delta_{u'}=\delta_{v'}=\delta_x+1$ and $\delta_{w}=\delta_{u}-1=\delta_v-1$, this becomes: 

$$\mathcal{S}(\widetilde{T})=\mathcal{S}(T)-\delta_u-\delta_u-\delta_x+(\delta_x+1)+(\delta_{x}+1)+ (\delta_{u}-1)$$ $$=\mathcal{S}(T)+\underbrace{\delta_x-\delta_u}_{\stackrel{\geq 0} {\mbox{\tiny as $x$ has max. depth}}}+1 > \mathcal{S}(T).$$

So $\mathcal{S}(\widetilde{T}) >\mathcal{S}(T) $. Clearly, this contradicts the maximality of $T$, which shows that the assumption was wrong. Thus, the caterpillar is the unique binary tree maximizing $\mathcal{S}$. This completes the first part of the proof.

\item Note that the caterpillar $T_n^{cat}$ on $n$ leaves has one leaf of depths 1 to $n-2$  each, but it has precisely two leaves of depth $n-1$, namely the two leaves in its only cherry. Thus, we have $\mathcal{S}(T_n^{cat}) = \sum\limits_{1}^{n-1}i + (n-1)$. Using the Gaussian sum, this immediately gives  $\mathcal{S}(T_n^{cat}) = \frac{(n-1)n}{2}+(n-1)=\frac{n\cdot(n+1)}{2}-1$. This completes the proof.
\end{enumerate}

\end{proof}

In total, we conclude that the caterpillar is indeed the unique binary tree maximizing the Sackin index, i.e. the caterpillar is the unique most imbalanced binary tree. So for all $n$, there is precisely one binary tree with maximal Sackin index. As we will show in Section \ref{sec_binmin}, the situation is entirely different for the minimal Sackin index in the binary case, as it can be taken on by various trees (depending on $n$). Moreover, we will show in Section \ref{sec_nonbin} that the binary caterpillar remains the unique tree maximizing the Sackin index even if we do not restrict the maximization to binary trees.

\par\vspace{0.5cm}
We conclude this subsection by noting that the sequence of maximal Sackin values $(a_n)_{n \in \mathbb{N}_{\geq 1}}$ with $a_n=\mathcal{S}(T^{cat}_n)=\frac{n(n+1)}{2}-1$ for $i \in \mathbb{N}\geq 1$, which starts with 2, 5, 9, 14, 20, 27, 35, 44, 54, 65, 77, 90, 104, 119, 135, 152, 170, 189, 209, $\ldots$, corresponds to sequence A000096 in the Online Encyclopedia of Integer Sequences OEIS \cite[Sequence A000096]{OEIS}, when the index is shifted by 1 (i.e. the $i$\textsuperscript{th} entry of the OEIS sequence corresponds to the $(i+1)$\textsuperscript{st} entry of our sequence). So this sequence has already occurred in other contexts, which might link the maximal Sackin index to other areas of research like the study of prime polyominoes or the traveling salesman polytope \cite[Sequence A000096]{OEIS}.

\subsubsection{Maximally balanced binary trees} \par

In this subsection, we first want to establish the same results for binary trees with minimal Sackin index that the previous section stated for trees with maximal Sackin index. In particular, we want the minimum value of the Sackin index and we want to characterize the trees that achieve it. However, it turns out that -- as opposed to the previous section -- the case of minimality is far more involved. \red{Some approaches concerning the minimum Sackin value can be found in the literature on rooted binary trees -- for instance in the appendix of \cite{shao}, where a somewhat complicated (and unfortunately erroneous) attempt at calculating this minimum value was made\footnote{\red{The definition of the function $F(t,i)$ presented in \cite{shao} should actually be called $F(t,i,j)$, as it depends on $j$, too, and for the last summand, i.e. for $j=m_i$, the definition of this function does not coincide with the (correct) example presented at the very end of that manuscript. Thus, the values calculated by Equation $(A1)$ in that manuscript do \emph{not} coincide with the true minimal Sackin values, e.g. for $n=11$.}}, albeit without a formal proof. 
Independently, \emph{ordered} binary trees with minimum internal path length have already been characterized in \cite[p. 401]{knuth1}, which can be used to derive a formula for the minimal possible Sackin index for $n$ leaves. We show how to do so in the proof of Theorem \ref{thm_twodepthsminNEW}.} 
\par \vspace{0.5cm}
The reason why the minimum is more involved than the maximum is that, depending on the number $n$ of leaves, the binary tree with minimal Sackin index need not be unique, so counting these trees is more complicated. Note that while examples for the fact that the Sackin index can be minimized by more than one tree have been presented before, e.g. in \cite{cophenetic}, see Figure \ref{6taxabalanced}, it has so far not been investigated for which values of $n$ this happens and precisely how many minima there are for each $n$. \red{It is the main aim of this subsection to characterize all binary Sackin minimal trees and to present a recursive approach to count them.}

\begin{figure}
	\centering
	\includegraphics[scale=.4]{sackin6taxamin.eps}
	\caption{Two trees with $n=6$ leaves which both have Sackin index 16, which can be verified to be minimal for $n=6$ by using Theorem \ref{thm_twodepthsminNEW}, as both of them only employ leaves of depths 2 and 3.
	}
	\label{6taxabalanced}
\end{figure}

\par \vspace{.5cm}
We now state  the first theorem of this subsection, which gives a full characterization of all rooted binary trees with minimum Sackin index and also provides an explicit formula to calculate this minimum value. \color{black} Note that a variation of the statement of this  theorem can also be found, for instance, in \cite{cameron}, where it is attributed to Knuth \cite{knuth3}. Indeed, as we will outline in the proof, Knuth stated the main ideas underlying this assertion, but not in \cite{knuth3}. Instead, these ideas can be found in \cite[pp. 400--401]{knuth1}.

\newpage
\begin{theorem}\label{thm_twodepthsminNEW}
Let $T$ be a rooted binary tree on $n$ leaves and let $k=\lceil \log_2(n)\rceil$. Then, $T$ minimizes the Sackin index amongst all such trees if and only if $T$ equals $T_k^{fb}$ (for $n=2^k$) or $T$ employs precisely two leaf depths, namely $k-1$ and $k$ (for $n <2^k$). Moreover, if $T$ has minimal Sackin index, we have $\mathcal{S}(T)=-2^k+n(k+1)$ (which equals $k\cdot 2^k$ if $n=2^k$).
\end{theorem}

\begin{proof} First, note that by Lemma \ref{lem_inout}, it suffices to analyze the total \emph{internal} path length $i(T)$ instead of $\mathcal{S}(T)$. And as explained already in \cite{knuth1}, $i(T)$ is clearly minimal if and only if every internal node is as close to the root as possible, but in a rooted binary tree, there are at most $2^m$ vertices at distance $m$ to the root for all $m=1,\ldots,k$. Now, as we have $n\in \{2^{k-1}+1,\ldots,2^k\}$ leaves, we know that (as $T$ is binary) we have $n-1$ inner vertices with $2^{k-1}\leq n-1 \leq 2^k-1$. This implies that in a tree which minimizes the total internal path length, the \emph{inner} vertices form a subtree $T'$ which consists of $T_{k-2}^{fb}$, which has $2^{k-1}-1$ vertices (and the \enquote{leaves} of which have distance $k-2$ from the root), with $(n-1)-(2^{k-1}-1)=n-2^{k-1}$ extra vertices attached at distance $k-1$ to the root. Note that some of the vertices at distance $k-2$ from the root in $T'$ may be degree-2 vertices. So if $T$ contains such a subtree $T'$ of inner vertices, the leaves of $T$ must be attached such that all degree-2 vertices of $T'$ are attached to one leaf of $T$, and all leaves of $T'$ are attached to two leaves of $T$ (i.e. leaves of $T'$ lead to cherries in $T$). Note that this implies that all leaf depths in $T$ are $k-2+1=k-1$ or $k-1+1=k$, respectively, and the number of leaves at distance $k$ equals $2\cdot n-2^{k-1}=2n-2^k$. The latter assertion is due to the fact that only the extra vertices of depth $k-1$ in $T'$ can lead to depth-$k$ vertices in $T$, but they are leaves in $T'$, so each of them must be the parent of a cherry in $T$ (which is why each extra vertex leads to two vertices of maximal depth in $T$). 

In case that $n=2^k$, by the reasoning above we have that there are $2n-2^k=2^{k+1}-2^k=2^k=n$ many leaves at distance $k$ to the root, so all leaves have distance $k$, which implies that $T$ equals $T_k^{fb}$.

It only remains to show that $\mathcal{S}(T)=-2^k+n(k+1)$. But this is now obvious by Definition \ref{sackinalternative}, as we now know that a tree with $n$ leaves and minimum Sackin index hhas $2n-2^k$ leaves at depth $k$ and $n-(2n-2^k)=2^k-n$ many leaves at depth $k-1$. This implies that $\mathcal{S}(T)=(2n-2^k)k + (2^k-n)(k-1)=-2^k+n(k+1)$. Note that in case that $n=2^k$, this implies $\mathcal{S}(T)=-2^k+n(k+1)=-2^k+2^k(k+1)=k\cdot 2^k$.This completes the proof.
\end{proof}

Before we continue, we briefly turn our attention to theoretical computer science.

\begin{remark} \label{rem_knuth_minima} In \cite[pp. 400--401]{knuth1}, the author proves an explicit formula for the minimum total \emph{internal} path length for rooted binary ordered trees with $m$ internal nodes, i.e. for the sum of all depths of the $m$ internal nodes: this number equals $(m+1)q-2^{q-1}+2$, where $q=\lfloor \log_2(m+1)\rfloor$. Using  Lemma \ref{lem_inout}, it can be easily shown that the value stated in \cite{knuth1} equals $\mathcal{S}(T)=-2^k+n(k+1)$, so this leads to an alternative proof for the minimum value stated in Theorem \ref{thm_twodepthsminNEW}. Moreover, it is stated there that the \enquote{optimum value ... is clearly achieved in a tree that looks like this}, followed by an example for $m=12$ and thus with $n=13$ and $k=4$. This example is a tree with 13 leaves, 3 of which have depth $3=k-1$ and 10 of which have depth $4=k$ -- so it is $T_3^{fb}$ in which 5 leaves have been replaced by cherries.

Note that while the author is explicitly referring to ordered trees \cite[p. 309]{knuth1}, it is obvious that the Sackin value is independent of the ordering. Thus, Theorem \ref{thm_twodepthsminNEW} can be derived from the groundbreaking work by Knuth in the described manner. However, as we will see in Section \ref{sec_ordered}, the ordering does indeed make quite a difference when the number of Sackin minimal trees shall be counted, as there are far more ordered ones than non-ordered ones (even though the ordered ones can be counted more easily). \end{remark}

\color{black}

Next, we consider the sequence of minimal Sackin values.

\begin{remark} Note that Theorem \ref{thm_twodepthsminNEW} implies that the sequence of minimal values of the Sackin index, starting at $n=1$, is $1,2,5,8,12,16,20,24,\ldots$ This corresponds to Sequence A003314 in the Online Encyclopedia of Integer Sequences OEIS \cite[Sequence A003314]{OEIS}, which is also often referred to as binary entropy function, which interestingly links the Sackin index to the areas of information theory \cite{shannon} and thermodynamics \cite{landauer}. 
\end{remark}

We will now turn our attention to the set of trees of minimal Sackin index, i.e. the maximally balanced ones according to this index. As we have seen in Theorem \ref{thm_twodepthsminNEW}, if the number of leaves is a power of 2, i.e. $n=2^k$, the maximally balanced tree is unique, namely $T_k^{bal}$. While this had previously been observed in the literature (cf. \citep{heard,shao}), no statement on the number of maximally balanced trees (ordered or not) for leaf numbers that are not a power of 2 has been made in the literature so far, even though it has long been known that need not be unique for all values of $n$. In fact, already for $n=6$, there are two minima, which are depicted in Figure \ref{6taxabalanced}. As stated before, this example is not new; it can for instance  already be found in \cite{cophenetic}. However, so far the explicit number of Sackin minima for $n \neq 2^k$ has not been investigated. In Theorem \ref{recursion}, we will provide a recursive formula to calculate this number in the following, and we will exploit the full characterization of Sackin minimal trees provided by Theorem \ref{thm_twodepthsminNEW} in order to do so.  

Counting the symmetries that may occur when some leaves of $T_{k-1}^{fb}$ are replaced by cherries is not trivial. So in order to conclude this section, we will show in the following that the number of trees with minimal Sackin index can be recursively counted. The proof of this formula exploits Theorem \ref{thm_twodepthsminNEW}. 

We now state the main result of this section.

\begin{theorem} \label{recursion} Let $s(n)$ denote the number of binary rooted trees with $n$ leaves and with minimal Sackin index and let $k = \lceil \log_2(n)\rceil$. In the following, we consider partitions of $n$ into two integers $n_a$, $n_b$, i.e. $n=n_a+n_b$. Moreover, let $f(n)=\begin{cases}0 & \mbox{if $n$ is odd} \\ {s(\frac{n}{2})+1 \choose 2}& \mbox{else. }\end{cases}$ \\ \par\vspace{0.3cm}Then, the following recursion holds: 
\begin{itemize}
\item $s(1)=1$
\item $s(n)=  \sum\limits_{\substack{(n_a,n_b): \\n_a+n_b=n,\\ \frac{n}{2}< n_a\leq 2^{k-1} ,\\ n_b\geq 2^{k-2} }}  s(n_a)\cdot s(n_b) +f(n)$ for all $n>1$.
\end{itemize}

\end{theorem}

\begin{proof} First consider $n=1$. In this case, it is clear that there is only one rooted binary tree, namely the one consisting of only one node, which therefore has minimal Sackin index, which implies $s(1)=1$.

Now let $n>1$. In this case, $n$ can be partitioned into two summands $n_a$ and $n_b$ such that $n_a \geq n_b$ and thus in particular $n_a\geq \frac{n}{2}$. We first consider the case where $n_a \neq n_b$. 

Note that a Sackin minimal rooted binary tree $T$ with at least 2 leaves can be decomposed into its two maximal pending subtrees $T_a$ and $T_b$, both of which must be Sackin minimal, too (otherwise we could replace the non-minimal one by a minimal one on the same number of leaves and thus achieve a tree with a lower Sackin index, which would contradict the minimality of $T$). So we need to count the number of such combinations, where $T_a$ has $n_a$ leaves and $T_b$ has $n_b$ many leaves. 

We first consider the trees for which $n_a > \frac{n}{2}$ and thus $n_a > n_b$. As we only consider Sackin minimal trees, we know by Theorem \ref{thm_twodepthsminNEW} that each such tree can be constructed by taking $T_{k-1}^{fb}$ and replacing $n-2^{k-1}$ leaves with cherries. Note that if in $T_b$ no leaves are replaced by cherries, $T_b$ equals tree $T_{k-2}^{fb}$, in which case $n_b=2^{k-2}$;  in all other cases, $n_b>2^{k-2}$. This explains why we only consider cases with $n_b \geq 2^{k-2}$. Again by considering Theorem \ref{thm_twodepthsminNEW}, we can conclude that if \emph{all} leaves in $T_a$ are replaced by cherries, $T_a$ equals $T_{k-1}^{fb}$ and thus has $2^{k-1}$ leaves. Otherwise, $T_a$ has fewer leaves. This leads to the requirement $n_a \leq 2^{k-1}$.
So all $s(n_a)\cdot s(n_b)$ such combinations of Sackin minimal trees on $n_a$ and $n_b$ leaves (where $\frac{n}{2} \leq n_a\leq 2^{k-1}$ and $n_b \geq 2^{k-2}$ and $n_a+n_b=n$, respectively) lead to different Sackin minimal trees on $n$ leaves (counting trees twice has been prevented by the restriction $n_a\geq \frac{n}{2}$). Clearly, this consideration recovers all trees with $n_a>n_b$ that have the structure described in  Theorem \ref{thm_twodepthsminNEW} and thus have minimal Sackin index. This explains the first summand in the recursive formula.

Now consider the case $n_a=n_b=\frac{n}{2}$, which of course only needs to be considered if $n$ is even (which explains why $f(n)=0$ if $n$ is odd). In this case, if we consider all $s(n_a)\cdot s(n_b)=s^2\left(\frac{n}{2}\right)$ combinations of Sackin minimal trees on $\frac{n}{2}$ leaves, then due to symmetry, the ones where $T_a$ and $T_b$ are not isomorphic will be counted twice, but not the ones where $T_a$ and $T_b$ are isomorphic. So we first have to add once more the $s \left(\frac{n}{2}\right)$ trees where $T_a$ and $T_b$ are isomorphic. So now we are in total considering $s^2\left(\frac{n}{2}\right)+s \left(\frac{n}{2}\right)$ many trees, but each tree now occurs two times. So the number of distinct Sackin minimal trees with $n_a=n_b$ equals $$\frac{1}{2}\cdot \left(s^2\left(\frac{n}{2}\right)+s \left(\frac{n}{2}\right)\right)={s \left(\frac{n}{2}\right)+1 \choose 2},$$ which explains the second summand in the recursion and thus completes the proof.
\end{proof}

We can use Theorem \ref{recursion} to derive the following corollary, which characterizes all cases in which the binary tree with minimal Sackin index is actually unique.  

\begin{corollary}\label{uniquesackinminima} Let $n\in \mathbb{N}$. Then, there is only one rooted binary tree $T$ with minimum Sackin index if and only if there exists an $m \in \mathbb{N}$ such that $n\in \{2^m-1,2^m,2^m+1\}$.
\end{corollary}

\begin{proof} In the following, let $n=n_a+n_b$ with $n_a \geq n_b$ and $k=\lceil\log_2(n)\rceil$. Now if $n=2^m$, then by Theorem \ref{thm_twodepthsminNEW} the minimum is indeed unique, which completes the first part of the proof.  

Next, consider the case $n=2^m-1$. Then, we prove the statement by induction on $m$. If $m=1$, we have $n=2^m-1=1$, which implies that there is only one rooted binary tree (namely the one consisting of only one node), so there remains nothing to show. So now assume that the statement holds for $m-1$ and consider $m$. Clearly, since $n=2^m-1$, we have $m=k$, and by the reasoning explained in the proof of Theorem \ref{recursion}, we must have $n_a \leq 2^{k-1}$. On the other hand, we also have $n_a \geq \frac{n}{2} = 2^{k-1}-\frac{1}{2}$ as $n$ is odd. As $n_a\in \mathbb{N}$, this implies $n_a \geq 2^{k-1}$. So altogether, we have $n_a=2^{k-1}$. This implies $n_b = n-n_a=2^k-1-2^{k-1}=2^{k-1}-1$. So $n_a$ and $n_b$ are uniquely determined, which implies that the sum in the recursion stated by Theorem \ref{recursion} only has one summand. This summand, however, is $s(n_a)\cdot s(n_b)=s(2^{k-1}) \cdot s(2^{k-1}-1)$, which by the first part of the proof and by induction equals $1 \cdot 1 =1$. Moreover, $f(n)=f(2^m-1)=0$ as $n$ is odd. This completes the second part of the proof.

Next, consider the case $n=2^m+1$. Then, we prove the statement by induction on $m$. If $m=1$, we have $n=2^m+1=3$, which implies that there is only one rooted binary tree (namely the rooted 3-leaf caterpillar), so there remains nothing to show. So now assume that the statement holds for $m-1$ and consider $m$. Clearly, since $n=2^m+1$, we have $k=m+1$, and by the reasoning explained in the proof of Theorem \ref{recursion}, we must have $n_a \leq 2^{k-1}=2^m$. On the other hand, we also have $n_a \geq \frac{n}{2} = 2^{m-1}+\frac{1}{2}$ as $n$ is odd. As $n_a\in \mathbb{N}$, this implies $n_a \geq 2^{m-1}+1$. Thus, $n_b \leq n-n_a=2^m+1 - (2^{m-1}+1) =2^{m-1}$. But as explained in the proof of Theorem \ref{recursion}, we must also have $n_b \geq 2^{k-2}=2^{m-1}$. So in total, $n_b=2^{m-1}$, which in turn implies that $n_a=n-n_b=2^m+1-2^{m-1}=2^{m-1}+1$. So $n_a$ and $n_b$ are uniquely determined, which implies that the sum in the recursion stated by Theorem \ref{recursion} only has one summand. This summand, however, is $s(n_a)\cdot s(n_b)=s(2^{m-1}+1) \cdot s(2^{m-1})$, which by the first part of the proof and by induction equals $1 \cdot 1 =1$. Moreover, $f(n)=f(2^m+1)=0$ as $n$ is odd. This completes the third part of the proof.

Now, assume that $n\notin \{2^m-1,2^m,2^m+1\}$ for any $m$. In particular, this implies that $n\geq 6$ and thus $k=\lceil\log_2(n)\rceil \geq 3$. In particular, for $k=\lceil\log_2(n)\rceil$, this implies that $n \in \{2^{k-1}+2,\ldots,2^k-2\}$. We first show that there are at least two choices for $n_b$ which both lead to valid choices for $n_a$. 
\begin{itemize} 
\item Consider $n_b:=2^{k-2}$ (which is possible as $k \geq 3$). This leads to $n_a=n-n_b \geq 2^{k-1}+2-2^{k-2} = 2^{k-2}+2 > n_b$ (and thus $n_a > \frac{n}{2}$) and $n_a=n-n_b \leq 2^k-2-2^{k-2} = 3\cdot 2^{k-2}-2<2^{k-1}$. So clearly, this choice of $n_b$ leads to a valid pair $(n_a,n_b)$ which will be considered in the sum of the recursion stated in Theorem \ref{recursion}.
\item Consider $n_b:=2^{k-2}+1$ (which is possible as $k \geq 3$). This leads to $n_a=n-n_b \geq 2^{k-1}+2-2^{k-2} -1= 2^{k-2}+1 > n_b$ (and thus $n_a > \frac{n}{2}$) and $n_a=n-n_b \leq 2^k-2-2^{k-2}-1 = 3\cdot 2^{k-2}-3<2^{k-1}$. So clearly, this choice of $n_b$ leads to a valid pair $(n_a,n_b)$ which will be considered in the sum of the recursion stated in Theorem \ref{recursion}.
\end{itemize}
So we have found two different summands for the recursion stated in Theorem \ref{recursion}, and as the recursion starts with $s(1)=1$, each pair contributes at least $s(n_a)\cdot s(n_b)=1 \cdot 1 =1$ to the sum. This immediately implies $s(n) \geq 2$.

So in summary, the minimum is unique if and only if $n\in \{2^m-1,2^m,2^m+1\}$ for some $m \in \mathbb{N}$. This completes the proof.
\end{proof}

\begin{remark} Starting at $n=1$ and continuing up to $n=32$, the sequence $s(n)$ of numbers of trees with $n$ leaves and with minimal Sackin index is 1, 1, 1, 1, 1, 2, 1, 1, 1, 3, 3, 5, 3, 3, 1, 1, 1, 4, 6, 14, 17, 27, 28, 35, 28, 27, 17, 14, 6, 4, 1, 1. We have calculated the values of $s(n)$ for up to $n=1024$. These data can be found online at \cite{NumberOfSackinMinima}. Note that this sequence is new to the Online Encyclopedia of Integer Sequences OEIS \cite[Sequence A299037]{OEIS}; it has been submitted in the scope of this manuscript. It had previously not been contained in the OEIS, i.e. this sequence has so far apparently not occurred in any other context.
\end{remark}

We end this section by noting that Theorem \ref{thm_twodepthsminNEW} actually guarantees that the difference between the two maximal pending subtrees $T_a$ and $T_b$ of the standard decomposition of a minimum Sackin tree does not get too large (in terms of the number of leaves). We will quantify this in the following corollary, which is is actually useful to investigate other balance indices like e.g. the Colless index \cite{colless}, and particularly their extremal properties \cite[Lemma 6]{Coronado2020}.\footnote{Corollary \ref{collesscorollary} corresponds to Corollary 4 in a preprint of the present manuscript, which was cited in \cite{Coronado2020}.}

\begin{corollary}\label{collesscorollary}
Let $T$ be a rooted binary tree with $n\in \mathbb{N}_{\geq 2}$ leaves. Moreover, let $T=(T_a,T_b)$ be the standard decomposition of $T$ into its two maximal pending subtrees, let $n_i$ denote the number of leaves in $T_i$ for $i\in \{a,b\}$, respectively, such that $n_a\geq n_b$. Let $k=\lceil \log_2 n\rceil$. Then, the following equivalence holds:
$T$ has minimum Sackin index if and only if $ T_a$ and $T_b$ have minimum Sackin index and $n_a-n_b\leq \min\{n-2^{k-1},2^k-n\}.$
\end{corollary}

\begin{proof}
We first consider the case where $T$ has minimum Sackin index. It is clear that $T_a$ and $T_b$ then also must have minimum Sackin index, because by Definition \ref{sackin}, it can be easily seen that $\mathcal{S}(T)=\mathcal{S}(T_a)+\mathcal{S}(T_b)+n$. Thus, if, say, $T_b$ was not minimal, we could replace $T_b$ in $T$ by a tree $T_b'$ on $n_b$ leaves such that $\mathcal{S}(T_b')<\mathcal{S}(T_b)$. This would turn $T$ into a new tree $T'$ on $n$ leaves with $\mathcal{S}(T')=\mathcal{S}(T_a)+\mathcal{S}(T_b')+n<\mathcal{S}(T_a)+\mathcal{S}(T_b)+n=\mathcal{S}(T)$, which would contradict the minimality of $\mathcal{S}(T)$. 

It remains to show that $n_a-n_b\leq \min\{n-2^{k-1},2^k-n\}$. In order to do so, recall that by the reasoning explained in the proof of Theorem \ref{recursion}, in a Sackin minimal tree we must have $n_b\geq 2^{k-2}$ and $n_a \leq 2^{k-1}$ (In the sum of the recursion, this is explicitly stated, but both statements also hold if $n_a=n_b=\frac{n}{2}$, because $2^{k-1}<n \leq 2^k$ by the definition of $k$).

Now first assume that $n_a-n_b>n-2^{k-1}=n_a+n_b-2^{k-1}$. This implies $-2n_b>-2^{k-1}$ and thus $n_b < 2^{k-2}$, a contradiction.

Next, assume $n_a-n_b>2^k-n=2^k-n_a-n_b$. This implies $2n_a> 2^k$ and thus $n_a > 2^{k-1}$, a contradiction.

So we must have $n_a-n_b\leq \min\{n-2^{k-1},2^k-n\}$. This completes the proof. 
\end{proof}

\color{black}
\section{Counting different types of extremal Sackin trees} \label{sec_transfer}

In this section, we first generalize our considerations to non-binary trees and then transfer our findings on the Sackin index back to the two disciplines where they are probably needed the most, namely mathematical phylogenetics and theoretical computer science. 

\subsection{Non-binary trees}\label{sec_nonbin}

Both in computer science as well as in mathematical phylogenetics, where tree balance plays an important role, binary trees are by far the most relevant trees. However, in both contexts, at times also non-binary trees\footnote{Recall that \enquote{non-binary} is used here as is common in the phylogenetic literature, i.e. in the sense of \enquote{not necessarily binary}.} play a role (for instance, in phylogenetics you can depict insecurity concerning the speciation order by non-binary trees). Also, as opposed to some other balance indices like e.g. the Colless index \cite{colless}, the definition of the Sackin index works for general rooted trees and not just for binary ones. Therefore, we drop the assumption of our trees being binary in this subsection and compare the extremal Sackin trees derived this way with the ones from the binary setting.

Again we start with the maximum Sackin index. In fact, we will show that nothing changes -- the binary caterpillar is still the unique tree maximizing the Sackin index even if the restriction to binary trees is dropped.

\begin{theorem}\label{thm_catNonBin} Let $T$ be a rooted non-binary tree with maximal Sackin index and $n$ leaves. Then, $T$ equals the caterpillar $T_n^{cat}$.
\end{theorem}

\begin{proof} Let $T$ be a non-binary tree with maximal Sackin index. Assume $T$ has an inner vertex $v$ that has at least three children $x$, $y$ and $z$ (and possibly more). We now construct a tree $T'$ as follows: We replace $v$ by two new vertices $v_1$ and $v_2$, which are connected by a new edge $(v_1,v_2)$, such that $v_1$ is adjacent to $x$ and $v_2$ and -- if $v$ is not the root of $T$ -- the ancestor of $v$ in $T$. Moreover, $v_2$ is adjacent to all other vertices to which $v$ is adjacent in $T$, that is, in particular, $y$ and $z$. So both $v_1$ and $v_2$ have at least two descendants. However, all leaves that descend from $v_2$ have increased their depth by 1 in $T'$ compared to $T$. This is due to the additional edge $(v_1,v_2)$ on the way from the root to these leaves. On the other hand, no depth has decreased in $T'$ compared to $T$. So by Definition \ref{sackinalternative}, this implies that the Sackin index of $T'$ is strictly larger than that of $T$, which contradicts the maximality of $T$. So the assumption was wrong, which implies that $T$ is binary. Thus, as $T$ is binary and has maximal Sackin index, by Theorem \ref{thm_sackinCat}  $T$ is a binary caterpillar. This completes the proof.
\end{proof}

\begin{remark} \label{rem_unbounded} Note that both for Theorem \ref{thm_sackinCat} as well as for Theorem \ref{thm_catNonBin} it is crucial that we excluded rooted trees with degree 2 vertices (other than possibly the root) from our considerations, because by adding more and more degree 2 vertices to a tree, we could increase the leaf depths and thus, by Definition \ref{sackinalternative}, the Sackin index would be unbounded.
\end{remark}

Now we turn our attention to the minimum Sackin index in the non-binary case. It turns out that interestingly, this scenario differs from the binary case, because while the minimum need not be unique in the binary case, in the non-binary case it is unique for all leaf numbers -- and even in the case where $n=2^k$ for some $k$, the unique non-binary tree minimizing the Sackin index is \emph{not} $T_k^{fb}$. In fact, the star tree turns out to be optimal for all $n$.

\begin{theorem} \label{thm_star} Let $n\geq 2$. Then, the star tree $T_n^{star}$ is the unique rooted non-binary tree minimizing the Sackin index, and we have $\mathcal{S}(T_n^{star})=n$.
\end{theorem}

\begin{proof} We first show the second assertion. Clearly, in a star tree, all leaves have distance 1 to the root, so $\mathcal{S}(T_n^{star})=\sum\limits_{x \in V^1}\delta_x=\sum\limits_{x \in V^1}1=n$. Moreover, this is clearly minimal, as in a tree with at least two leaves, we have $\delta_x\geq 1$ for all leaves. So the only thing left to show is uniqueness. Assume that $T$ is a rooted tree on $n$ leaves that also has minimal Sackin index but is not a star tree. Then $T$ has at least two inner vertices (the root $\rho$ and at least one more). Let $v$ be a child of $\rho$ that is also an inner vertex. Now we construct a tree $T'$ from $T$ by contracting the edge $(\rho,v)$. This way, all leaves descendant from $v$ now have strictly decreased their depths by 1, and no leaf has increased its depth. Thus, the Sackin index of $T'$ is strictly smaller than that of $T$, clearly contradicting the minimality of $T$. Thus, the assumption was wrong. This completes the proof of the theorem.
\end{proof}

In the following section, we turn our attention to mathematical phylogenetics, i.e. to leaf-labelled trees. 

\subsection{Phlyogenetic trees}\label{sec_phylo}

As mentioned earlier, trees play a fundamental role in mathematical phylogenetics. In fact, trees can be used to depict the evolutionary relationships between present-day species. However, mere graph theoretic trees as discussed so far in this manuscript do not suffice to do so; instead, we need \emph{phylogenetic trees} as introduced in Section \ref{sec_prelim}. 

We know from Theorems \ref{thm_sackinCat} and \ref{thm_catNonBin} that there is only one (graph theoretic) tree maximizing the Sackin index for any value of $n$ both in the binary and non-binary case (namely the caterpillar), but the story is quite different for phylogenetic trees: For instance, already for $n=3$, there are three different phylogenetic caterpillars (the one with the cherry $[1,2]$, the one with the cherry $[1,3]$ and the one with the cherry $[2,3]$), so the maximum need not be unique anymore. Concerning minimal trees, we know from Theorem \ref{recursion} how to calculate the number $s(n)$ of Sackin minimal (graph theoretical) trees, but we do not know yet how this generally translates to phylogenetic trees, as it is known that, due to symmetry, not all trees with $n$ leaves lead to the same number of phylogenetic trees. In fact, as a consequence of the famous Burnside Lemma \cite{burnside,frobenius}, the following assertion is well-known in mathematical phylogenetics.

\begin{proposition}{\upshape \textbf{(Corollary 2.4.3 in \cite{semple_steel_2003})}}\label{burnside} Let $T$ be a rooted binary tree with $n$ leaves. Let $\widetilde{s}(T)$ denote the number of symmetry nodes of $T$, i.e. the number of inner vertices of $T$ whose maximal pending subtrees are isomorphic. Then, the number of phylogenetic trees on leaf set $X=\{1,\ldots,n\}$ of shape $T$ equals $$\frac{n!}{2^{\widetilde{s}(T)}}.$$
\end{proposition}

This immediately leads to the following corollary for Sackin maximal phylogenetic trees.

\begin{corollary} Let $p^{max}(n)$ denote the number of phylogenetic $X$-trees with $X=\{1,\ldots,n\}$  and with maximal Sackin index. Then, we have $p^{max}(n)=\frac{n!}{2}$, and all of these trees have a binary caterpillar as underlying graph theoretic tree. This statement is still true even if we do not restrict the maximization to binary trees.
\end{corollary}

\begin{proof} As the Sackin index does not depend on the leaf labeling, it is clear by Theorems \ref{thm_sackinCat} and \ref{thm_catNonBin} that the underlying graph theoretic (i.e. non-leaf-labelled) tree must be a caterpillar. Moreover, a caterpillar has precisely one symmetry node, namely the parent of its only cherry. Thus, the assertion immediately follows by Proposition \ref{burnside}.
\end{proof}

\begin{remark} Starting at $n=1$ and continuing up to $n=20$, the sequence $p^{max}(n)$ of numbers of phylogenetic trees with $n$ leaves and with maximal Sackin index is 1, 3, 12, 60, 360, 2520, 20160, 181440, 1814400, 19958400, 239500800, 
3113510400, 43589145600, 653837184000, 10461394944000, 
177843714048000, 3201186852864000, 60822550204416000, 
1216451004088320000. Note that this sequence is well-known and plays a role in several contexts in combinatorics, for instance it describes the number of permutations of the numbers $1,\ldots,n$ in which $2$ follows $1$. For more details and other examples, we refer the interested reader to sequence A001710 in the Online Encyclopedia of Integer Sequences OEIS \cite{OEIS}.
\end{remark}

So as Proposition \ref{burnside} shows, symmetry nodes play a fundamental role in order to determine the number of phylogenetic trees per graph theoretic tree. For instance, tree $T_1$ from Figure \ref{6taxabalanced} has 3 symmetry nodes and six leaves, so there are $\frac{6!}{2^{3}}=90$ phylogenetic trees associated with $T_1$. But $T_2$, which has 4 symmetry nodes, only leads to $\frac{6!}{2^{4}}=45$ phylogenetic trees. So both trees lead to different numbers of phylogenetic trees, but both trees form the set of maximally balanced graph theoretic trees according to the Sackin index ($s(6)=2$ and we have already seen that $T_1$ and $T_2$ are both maximally balanced). So in order to determine the number of maximally balanced phylogenetic trees, we cannot simply multiply $s(T)$ with some number depending on $n$, because in fact, the number of associated phylogenetic trees might be different for all trees counted by $s(T)$. However, we can use similar arguments as in Theorem \ref{recursion} to derive a recursive formula for the number of binary phylogenetic trees that minimize the Sackin index.

\begin{theorem}\label{phylorecursion} 
Let $p^{min}(n)$ denote the number of rooted binary phylogenetic $X$-trees with $X=\{1,\ldots,n\}$ and with minimal Sackin index, and let $k = \lceil \log_2(n)\rceil$. For any partition of $n$ into two integers $n_a$, $n_b$, i.e. $n=n_a+n_b$, we use $k_a$ and $k_b$ to denote $\lceil \log_2(n_a)\rceil$ and $\lceil \log_2(n_b)\rceil=\lceil \log_2(n-n_a)\rceil$, respectively. Moreover, let 

$$g(n)=\begin{cases}0 & \mbox{if $n$ is odd,} \\ {\frac{1}{2}\cdot {n \choose \frac{n}{2}} \cdot \left(p^{min}\left(\frac{n}{2}\right)\right)^2 }& \mbox{else.  }\end{cases} $$  Then, the following recursion holds: 
\begin{itemize}
\item $p^{min}(1)=1$
\item $p^{min}(n)=  \sum\limits_{\substack{(n_a,n_b): \\n_a+n_b=n,\\ \frac{n}{2}< n_a\leq 2^{k-1} ,\\ n_b\geq 2^{k-2} }}  {n \choose n_a} \cdot p^{min}(n_a)\cdot p^{min}(n_b) +g(n)$
\end{itemize}

\end{theorem}

\begin{proof} First consider $n=1$. In this case, it is clear that there is only one rooted binary phylogenetic tree, namely the one consisting of only one node labelled with 1, which therefore has minimal Sackin index, which implies $p^{min}(1)=1$.

Now let $n>1$ and recall the proof of Theorem \ref{recursion}: As explained there, $n$ can be partitioned into two summands $n_a$ and $n_b$ such that $n_a \geq n_b$ and thus in particular $n_a\geq \frac{n}{2}$. We first consider the case where $n_a \neq n_b$. 

Again as in the proof of Theorem \ref{recursion}, a Sackin minimal rooted binary phylogenetic tree $T$ with at least 2 leaves can be decomposed into its two maximal pending subtrees $T_a$ and $T_b$, both of which must be Sackin minimal, too (otherwise we could replace the non-minimal one by a minimal one on the same number of leaves and thus achieve a tree with a lower Sackin index, which would contradict the minimality of $T$). So we need to count the number of such combinations, where $T_a$ has $n_a$ leaves and $T_b$ has $n_b$ many leaves. 

We first consider the trees for which $n_a > \frac{n}{2}$ and thus $n_a > n_b$. As we only consider Sackin minimal phylogenetic trees and as the Sackin index depends only on the underlying graph theoretic tree and not on the leaf labelling, we can conclude by the same arguments as in the proof of Theorem \ref{recursion} that we only need to consider cases with $n_b \geq 2^{k-2}$ and $n_a \leq 2^{k-1}$, as no other tree can be Sackin minimal.

So all $p^{min}(n_a)\cdot p^{min}(n_b)$ such combinations of Sackin minimal rooted binary phylogenetic trees on $n_a$ and $n_b$ leaves (where $\frac{n}{2} \leq n_a\leq 2^{k-1}$ and $n_b \geq 2^{k-2}$ and $n_a+n_b=n$, respectively) lead to different Sackin minimal phylogenetic trees on $n$ leaves (counting trees twice has been prevented by the restriction $n_a\geq \frac{n}{2}$). However, we need to multiply these options with all possible ways to label $T_a$ and $T_b$: We have ${n \choose n_a}$ many ways to pick $n_a$ labels from our label set of size $n$ in order to assign these labels to $T_a$. The remaining labels get assigned to $T_b$. In summary, this explains the first summand in the recursion.

Now consider the case $n_a=n_b=\frac{n}{2}$, which of course only needs to be considered if $n$ is even (which explains why $g(n)=0$ if $n$ is odd). 

Now if $n$ is even, there is another summand: We can proceed as in Case 1, but here we have to consider all combinations of 2 trees from the $p(n_a)=p(n_b)$ phylogenetic trees with $n_a=n_b$ leaves and with minimum Sackin index, and there are $p(n_a)\cdot p(n_a) = p\left(\frac{n}{2} \right) \cdot p\left(\frac{n}{2} \right)$ such ordered combinations. However, note that this implies that we now have counted some combinations twice if you disregard the ordering; namely all those for which $T_a$ and $T_b$ are not isomorphic (and thus the root $\rho$ of $T$ is not a symmetry node), but not the ones for which $T_a$ and $T_b$ are isomorphic. We will come back to this fact in a bit. Now first note that we still need to consider all possible leaf labelings, of which there are ${n \choose n_a}= {n \choose \frac{n}{2}}$. However, if we label the trees like this, every tree that has non-isomorphic $T_a$ and $T_b$ subtrees will lead to different phylogenetic trees for each choice of leaf labels. So of every phylogenetic tree resulting from such a graph theoretic tree, we now have two copies in our multiset. For the trees with isomorphic $T_a$ and $T_b$, however, of which we only have one graph theoretic copy in the set, we have now also two phylogenetic trees each, because for instance the assignment of $\{1,\ldots,\frac{n}{2}\}$ to $T_a$ and $\{\frac{n}{2}+1,\ldots,n\}$ to $T_b$ results in precisely the same trees as the assignment of $\{\frac{n}{2}+1,\ldots,n\}$ to $T_a$ and $\{1,\ldots,\frac{n}{2}\}$ to $T_b$. In summary, we now have a multiset containing ${n \choose \frac{n}{2}} \cdot p\left(\frac{n}{2}\right) \cdot p\left(\frac{n}{2}\right)$ many phylogenetic trees, but each of them appears twice. So in order to count every phylogenetic tree only once, we have to multiply their number by $\frac{1}{2}$. This explains the second summand in the recursion and thus completes the proof.
\end{proof}

\begin{remark} Starting at $n=1$ and continuing up to $n=20$, the sequence $p^{min}(n)$ of numbers of rooted binary phylogenetic trees with $n$ leaves and with minimal Sackin index is 1, 1, 3, 3, 30, 135, 315, 315, 11340, 198450, 2182950, 16372125,
85135050, 297972675, 638512875, 638512875, 86837751000,
5861548192500, 259861969867500, 8445514020693750. We have calculated the values of $p(n)$ for up to $n=100$. These data can be found online at \cite{NumberOfPhyloSackinMinima}. Note that this sequence is so far not contained in the  Online Encyclopedia of Integer Sequences OEIS \cite{OEIS}, but it has recently been submitted and will soon appear online. As it had previously not been contained in the OEIS, this sequence has so far apparently not occurred in any other context.
\end{remark}

Before we end this subsection, we turn our attention again to the non-binary case and the minimal Sackin index. By Theorem \ref{thm_star}, the star tree is the unique non-binary tree with minimal Sackin index. This directly leads to the following corollary.

\begin{corollary} Let $n\geq 2$. Then, the unique rooted non-binary phylogenetic tree on leaf set $X=\{1,\ldots,n\}$ is $T^{star}_n$ with its leaves bijectively labelled by $X$.
\end{corollary}

\begin{proof} By Theorem \ref{thm_star}, it is clear that the underlying graph theoretic tree of a Sackin minimal rooted non-binary tree necessarily is $T_n^{star}$. Moreover, it can be easily seen (formally: by using the Burnside lemma, where all leaves are in the same equivalence class) that all permutations of leaf labels lead to the same phylogenetic tree due to symmetry. This completes the proof.
\end{proof}

\subsection{Ordered trees}\label{sec_ordered}

We now turn our attention to computer science, where so-called search trees play a fundamental role. Such trees basically form a kind of data structure which can be used to locate certain keys from within a set. However, just as in the previous subsection, mere graph theoretic trees are not sufficient to operate as search trees. This is due to the fact that for instance all keys to the left of an inner node with key $\mathcal{K}$ must be smaller than $\mathcal{K}$, whereas the keys to the right of this inner node must be larger than $\mathcal{K}$ (or vice versa, but the convention that the smaller keys are in the left subtree is quite common) \cite{knuth1,knuth3}. But this immediately implies that we need to be able to distinguish between left and right, which we did not do in the graph theoretic setting. This is why in computer science, rooted binary \emph{ordered} trees as introduced in Section \ref{sec_prelim} play a fundamental role. However, as we do not limit our analyses to the binary case, the example with left and right needs to be thought of in a more general way, accordingly.

\begin{remark} \label{rem_deg2} In computer science, as opposed to the settings discussed so far in this manuscript, inner vertices of total degree 2 are sometimes allowed, i.e. in such settings may happen that a vertex only has a left or only a right child. However, by the same reasoning as mentioned in Remark \ref{rem_unbounded}, this would mean that the maximimum of the Sackin index would be unbounded. Moreover, it can be easily seen that the set of Sackin minimal trees will not change if we allow for the optimization to consider trees with degree-2 vertices, too, because a tree with a degree-2 vertex can never have minimal Sackin index. This is due to the fact that suppressing a degree-2 vertex would strictly decrease the leaf depths of all leaves descending from this vertex. As degree-2 vertices thus do not contribute anything to the understanding of Sackin minimal ordered trees, we keep excluding them from our considerations, i.e. we still only consider trees without such vertices. \end{remark}

We now recall the following result from the literature.

\begin{lemma}[c.f. Lemma 2.6 in \cite{raaz}] \label{raazlem} Let $T$ be a rooted binary graph theoretic tree with $n\geq 2$ leaves. Then, the number of ordered trees corresponding to $T$ is $2^{n-1-\widetilde{s}(T)}$, where $\widetilde{s}(T)$ denotes the number of symmetry nodes of $T$, i.e. the number of inner nodes whose two maximal pending subtrees are isomorphic. 
\end{lemma}

The idea of Lemma \ref{raazlem} is that each of the $n-1$ internal nodes in a rooted binary tree except for the symmetry nodes doubles the number of orderings. The fact that symmetry nodes do not do this is due to the fact that isomorphic trees are undistinguishable and thus can only be counted once. 

Lemma \ref{raazlem} together with Theorem \ref{thm_sackinCat} immediately lead to the following result.

\begin{corollary} \label{orderedSackinmax}The number of binary rooted and ordered trees with $n$ leaves with maximal Sackin index is $2^{n-2}$ if $n\geq 2$ or 1 otherwise, and all of these trees have $T_n^{cat}$ as underlying graph theoretic tree. This statement is still correct in the non-binary case, i.e. even if drop the restriction of the maximization to binary trees.
\end{corollary}

\begin{proof} As the Sackin index does not depend on the leaf labeling, it is clear by Theorems \ref{thm_sackinCat} and \ref{thm_catNonBin} that the caterpillar is the only tree that can underly an ordered tree that maximizes the Sackin index, binary or non-binary. Moreover, a caterpillar has precisely one symmetry node, namely the parent of its only cherry. Thus, the assertion immediately follows by Lemma \ref{raazlem}.
\end{proof}

\begin{remark} Starting at $n=1$ and continuing up to $n=32$, the sequence of numbers of ordered trees with $n$ leaves and with maximal Sackin index is 1,1, 2, 4, 8, 16, 32, 64, 128, 256, 512, 1024, 2048, 4096, 8192, 16384, 32768, 65536, 131072, 262144, 524288, 1048576, 2097152, 4194304,8388608, 16777216, 33554432, 67108864,134217728, 268435456, 536870912, 1073741824. This sequence has already long been contained in the Online Encyclopedia of Integer Sequences OEIS \cite{OEIS} as sequence A011782. Note that this sequence is well-known and plays a role in several contexts in combinatorics, for instance in order to count so-called unimodal permutations of $n$ items. For more details and other examples, we refer the interested reader to sequence A011782 in the OEIS.
\end{remark}

One of the main reasons why tree balance plays a fundamental role in theoretical computer science is that if ordered search trees are reasonably balanced, they lead to an efficient search time \cite{knuth1}. Therefore, the number of maximally balanced trees, i.e. of trees with minimum Sackin index, is of far greater interest than the number stated by Corollary \ref{orderedSackinmax}. Thus, we now proceed with considering the number of such minima in the binary case, which can very easily be calculated.

\begin{theorem}\label{orderedformula} 
Let $ot(n)$ denote the number of rooted binary ordered trees with $n$ leaves and with minimal Sackin index, and let $k = \lceil \log_2(n)\rceil$. Then, we have: $$ot(n)={2^{k-1} \choose n-2^{k-1}}.$$
\end{theorem}

\begin{proof} As we have seen in Theorem \ref{thm_twodepthsminNEW}, all maximally balanced trees can be constructed by taking $T_{k-1}^{fb}$, i.e. a fully balanced tree of height $k-1$, and replacing $n-2^{k-1}$ leaves by cherries. There are ${2^{k-1} \choose n-2^{k-1}}$ ways to choose $n-2^{k-1}$ leaves from the $2^{k-1}$ leaves of $T_{k-1}^{fb}$ to turn them into cherries (unordered sampling without replacement). As we are considering ordered trees here, every such choice leads to a different tree. This completes the proof. \end{proof}

\begin{remark} Starting at $n=1$ and continuing up to $n=20$, the sequence $ot(n)$ of numbers of ordered trees with $n$ leaves and with minimal Sackin index is 1, 1, 2, 1, 4, 6, 4, 1, 8, 28, 56, 70, 56, 28, 8, 1, 16, 120, 560, 
1820, 4368, 8008, 11440, 12870, 11440, 8008, 4368, 1820, 560, 120, 16,
 1. We have calculated the values of $ot(n)$ for up to $n=256$. These data can be found online at \cite{NumberOfOrderedSackinMinima}. Note that this sequence is so far not contained in the  Online Encyclopedia of Integer Sequences OEIS \cite{OEIS}, but it has recently been submitted and will soon appear online. As it had previously not been contained in the OEIS, this sequence has so far apparently not occurred in any other context.
\end{remark}

Last, we once more turn our attention to the non-binary case and the minimal Sackin index. By Theorem \ref{thm_star}, the star tree is the unique non-binary tree with minimal Sackin index. This directly leads to the following corollary.

\begin{corollary} Let $n\geq 2$. Then, the unique rooted non-binary ordered tree with $n$ leaves is $T^{star}_n$ (regarded as ordered tree).
\end{corollary}

\begin{proof} By Theorem \ref{thm_star}, it is clear that the underlying graph theoretic tree of a Sackin minimal rooted non-binary tree necessarily is $T_n^{star}$. Moreover, it can be easily seen that due to symmetry, there is only one ordering of the star tree (as all leaves pendant from the root can be considered as isomorphic subtrees which are thus indistinguishable). This completes the proof.
\end{proof}

\section{Discussion}
One aim of this manuscript was to provide deep insight into Sackin minimal and maximal trees both in the binary and non-binary settings -- partially by proving new results and partially by transferring results from other research areas. The main aim, however, was to provide a way to count Sackin minimal trees in the binary case (as the maxima in the binary and non-binary cases as well as the minimum in the non-binary case are all unique, so there is nothing to count). This was achieved through the recursion given by Theorem \ref{recursion}. 

An additional aim of this manuscript was to count different types of Sackin minimal trees, too, and thus to make these results interesting and relevant for the very disciplines where balanced trees are most important: mathematical phylogenetics and theoretical computer science. This led to a recursive formula for rooted binary phylogenetic trees in Theorem \ref{phylorecursion} and an explicit formula for rooted binary ordered trees in Theorem \ref{orderedformula}. We also showed that in all other cases, the extremal Sackin trees are unique, so there is nothing to count. 

Our enumerations of trees with $n$ leaves have led to some sequences which are new to the Online Encyclopedia of Integer Sequences and some which are well-known and have occurred in other areas of combinatorics before. This link between tree balance and other  topics might inspire future research. 

Another possible area of future research is the relatedness of the Sackin index to other balance indices. It was recently proven that the set of Sackin minimal binary trees contains all Colless minimal and Cophenetic minimal trees \cite{Coronado2020}, for instance, but there are various other balance indices whose relationship to the Sackin index has not been investigated yet. 
Possibly the most important aspect for future research resulting from this manuscript, however, might be implications of our findings concerning the number of extremal trees concerning the Sackin index, on evolutionary models and their induced probability distributions on the tree space.

\section*{Acknowledgements}
The author wishes to thank Lina Herbst, Kristina Wicke and Mike Steel for very helpful discussions on the general topic of the Sackin index and for comments concerning an earlier version of the manuscript. Moreover, the author wishes to thank an anonymous reviewer for very helpful comments; particularly for very helpful comments which greatly helped to improve this manuscript. Last but not least, the author wishes to thank the joint research project {\bf \emph{DIG-IT!}} supported by the European Social Fund (ESF), reference: ESF/14-BM-A55-0017/19, and the Ministry of Education, Science and Culture of Mecklenburg-Vorpommern, Germany.

\color{black}

\bibliographystyle{natbib}
\bibliography{bibfile}

\begin{thebibliography}{}

\bibitem[Blum and Francois(2005)Blum and Francois]{blum}
Blum, M. and Francois, O. 2005.
\newblock On statistical tests of phylogenetic tree imbalance: The {S}ackin and
  other indices revisited.
\newblock {\em Mathematical Biosciences\/}, {195}(2): 141 -- 153.

\bibitem[Burnside(1897)Burnside]{burnside}
Burnside, W. 1897.
\newblock {\em Theory of groups of finite order\/}.
\newblock Cambridge University Press.

\bibitem[Cameron and Wood(1994)Cameron and Wood]{cameron}
Cameron, H. and Wood, D. 1994.
\newblock Maximal path length of binary trees.
\newblock {\em Discrete Applied Mathematics\/}, {55}(1): 15 -- 35.

\bibitem[Cleary {\em et~al.}(2015)Cleary, Fischer, Griffiths, and
  Sainudiin]{raaz}
Cleary, S., Fischer, M., Griffiths, R., and Sainudiin, R. 2015.
\newblock Some distributions on finite rooted binary trees.
\newblock \url{http://lamastex.org/preprints/20151231_SomeDistsFRBTrees.pdf}.

\bibitem[Colless(1982)Colless]{colless}
Colless, D. 1982.
\newblock Review of "{P}hylogenetics: the theory and practice of phylogenetic
  systematics".
\newblock {\em Systematic Zoology\/}, {31}: 100.

\bibitem[Coronado {\em et~al.}(2020)Coronado, Fischer, Herbst, Rossell{\'o},
  and Wicke]{Coronado2020}
Coronado, T.~M., Fischer, M., Herbst, L., Rossell{\'o}, F., and Wicke, K. 2020.
\newblock On the minimum value of the colless index and the bifurcating trees
  that achieve it.
\newblock {\em Journal of Mathematical Biology\/}, {80}(7): 1993--2054.

\bibitem[Felsenstein(2004)Felsenstein]{felsenstein_2004}
Felsenstein, J. 2004.
\newblock {\em Inferring phylogenies.}
\newblock Sinauer Associates, Massachusetts.

\bibitem[Fischer(2018)Fischer]{NumberOfSackinMinima}
Fischer, M. 2018.
\newblock Number of rooted binary trees with $n\leq 1024$ leaves and minimal
  {S}ackin index.
\newblock \url{http://mareikefischer.de/SupplementaryMaterial/Sackin.txt}.

\bibitem[Fischer(2020a)Fischer]{NumberOfOrderedSackinMinima}
Fischer, M. 2020a.
\newblock Number of rooted binary ordered trees with $n\leq 128$ leaves and
  minimal {S}ackin index.
\newblock
  \url{http://mareikefischer.de/SupplementaryMaterial/SackinOrdered.txt}.

\bibitem[Fischer(2020b)Fischer]{NumberOfPhyloSackinMinima}
Fischer, M. 2020b.
\newblock Number of rooted binary phylogenetic trees with $n\leq 100$ leaves
  and minimal {S}ackin index.
\newblock \url{http://mareikefischer.de/SupplementaryMaterial/SackinPhylo.txt}.

\bibitem[Fischer(2020c)Fischer]{mathematicapackage}
Fischer, M. 2020c.
\newblock Sackin{M}inimizer -- a {M}athematica package for calculating binary
  {S}ackin minimal trees with and without ordering.
\newblock \url{http://mareikefischer.de/Software/SackinMinimizer.m}.

\bibitem[Frobenius(1887)Frobenius]{frobenius}
Frobenius, G. 1887.
\newblock \"uber die congruenz nach einem aus zwei endlichen gruppen gebildeten
  doppelmodul.
\newblock {\em Journal f\"ur die reine und angewandte Mathematik\/}, {101}:
  273--299.

\bibitem[Heard(1992)Heard]{heard}
Heard, S. 1992.
\newblock Patterns in tree balance among cladistic, phonetic, and randomly
  generated phylogenetic trees.
\newblock {\em Evolution\/}, {46}(6): 1818--1826.

\bibitem[Inc.(2017)Inc.]{Mathematica}
Inc., W.~R. 2017.
\newblock Mathematica, {V}ersion 10.3.
\newblock Champaign, IL, 2017.

\bibitem[Klein and Wood(1989)Klein and Wood]{klein}
Klein, R. and Wood, D. 1989.
\newblock The path length of binary trees.
\newblock In W.~Litwin and H.-J. Schek, editors, {\em Foundations of Data
  Organization and Algorithms\/}, pages 128--136, Berlin, Heidelberg. Springer
  Berlin Heidelberg.

\bibitem[Knuth(1997)Knuth]{knuth1}
Knuth, D. 1997.
\newblock {\em The art of computer programming: fundamental algorithms\/},
  volume~1.
\newblock Addison-Wesley, 3 edition.

\bibitem[Knuth(1998)Knuth]{knuth3}
Knuth, D. 1998.
\newblock {\em The art of computer programming: sorting and searching\/},
  volume~3.
\newblock Addison-Wesley, 2 edition.

\bibitem[{Landauer}(1992){Landauer}]{landauer}
{Landauer}, R. 1992.
\newblock Information is physical.
\newblock In {\em Workshop on Physics and Computation\/}, pages 1--4.

\bibitem[M.~Coronado {\em et~al.}(2020)M.~Coronado, Mir, Rossell{\'o}, and
  Rotger]{spanierNeu}
M.~Coronado, T., Mir, A., Rossell{\'o}, F., and Rotger, L. 2020.
\newblock On sackin's original proposal: the variance of the leaves' depths as
  a phylogenetic balance index.
\newblock {\em BMC Bioinformatics\/}, {21}(1): 154.

\bibitem[Mir {\em et~al.}(2013)Mir, Rossello, and Rotger]{cophenetic}
Mir, A., Rossello, F., and Rotger, L. 2013.
\newblock A new balance index for phylogenetic trees.
\newblock {\em Mathematical Biosciences\/}, {241}(1): 125 -- 136.

\bibitem[Sackin(1972)Sackin]{sackin}
Sackin, M. 1972.
\newblock "good" and "bad" phenograms.
\newblock {\em Systematic Zoology\/}, {21}: 225.

\bibitem[Semple and Steel(2003)Semple and Steel]{semple_steel_2003}
Semple, C. and Steel, M. 2003.
\newblock {\em Phylogenetics.}
\newblock Oxford University Press.

\bibitem[{Shannon}(1948){Shannon}]{shannon}
{Shannon}, C.~E. 1948.
\newblock A mathematical theory of communication.
\newblock {\em The Bell System Technical Journal\/}, {27}(3): 379--423.

\bibitem[Shao and Sokal(1990)Shao and Sokal]{shao}
Shao, K.-T. and Sokal, R. 1990.
\newblock Tree balance.
\newblock {\em Systematic Zoology\/}, {39}(3): 266--276.

\bibitem[Sloane(2018)Sloane]{OEIS}
Sloane, N. 2018.
\newblock The {O}n-{L}ine {E}ncyclopedia of {I}nteger {S}equences {OEIS}.
\newblock \url{https://oeis.org}.

\bibitem[Steel(2016)Steel]{steel2016}
Steel, M. 2016.
\newblock {\em Phylogeny: Discrete and random processes in evolution\/}.
\newblock CBMS-NSF Regional conference series in Applied Mathematics. SIAM.

\bibitem[Than and Rosenberg(2014)Than and Rosenberg]{noah}
Than, C. and Rosenberg, N. 2014.
\newblock Mean deep coalescence cost under exchangeable probability
  distributions.
\newblock {\em Discrete Applied Mathematics\/}, {174}: 11--26.

\bibitem[Wong and Nievergelt(1973)Wong and Nievergelt]{nievergelt}
Wong, C.~K. and Nievergelt, J. 1973.
\newblock Upper bounds for the total path length of binary trees.
\newblock {\em J. ACM\/}, {20}(1): 1–6.

\end{thebibliography}

\end{document}